\documentclass[journal]{IEEEtran}
\usepackage{cite}
\ifCLASSINFOpdf
\else
   \usepackage[dvips]{graphicx}
\fi

\newcommand{\cC}{{\mathcal{C}}}

\newcommand{\cE}{{\mathcal{E}}}

\newcommand{\cN}{{\mathcal{N}}}

\newcommand{\cP}{{\mathcal{P}}}

\newcommand{\cS}{{\mathcal{S}}}

\newcommand{\cT}{{\mathcal{T}}}

\newcommand{\cX}{{\mathcal{X}}}

\newcommand{\al}{\alpha}

\newcommand{\eps}{\varepsilon}

\usepackage{verbatim} % get comment environment, + new verbatim
 \usepackage{amsxtra}  % get, eg, \accentedsymbol
\newcounter{actr}
{\begin{list}{(\alph{actr})}{\usecounter{actr}}}{\end{list}}

\newcounter{ictr}
{\begin{list}{(\roman{ictr})}{\usecounter{ictr}}}{\end{list}}

\newtheorem{thm}{Theorem}
\newtheorem{lemma}{Lemma}

\newtheorem{corol}{Corollary}
\newtheorem{prop}{Proposition}

\newenvironment{new-proof}[1]
{{\em Proof  #1: }}%
{ \noindent\qed }
\newcommand{\qed}{\rule[0.1ex]{1.4ex}{1.6ex}}

\newcommand{\mrm}{\mathrm}

\DeclareMathAlphabet{\mathbsf}{OT1}{cmss}{bx}{n}% bold sans serif
\DeclareMathAlphabet{\mathssf}{OT1}{cmss}{m}{sl}% slanted sans serif

\DeclareSymbolFont{bsfletters}{OT1}{cmss}{bx}{n}  
\DeclareSymbolFont{ssfletters}{OT1}{cmss}{m}{n}
\DeclareMathSymbol{\bsfGamma}{0}{bsfletters}{'000}
\DeclareMathSymbol{\ssfGamma}{0}{ssfletters}{'000}
\DeclareMathSymbol{\bsfDelta}{0}{bsfletters}{'001}
\DeclareMathSymbol{\ssfDelta}{0}{ssfletters}{'001}
\DeclareMathSymbol{\bsfTheta}{0}{bsfletters}{'002}
\DeclareMathSymbol{\ssfTheta}{0}{ssfletters}{'002}
\DeclareMathSymbol{\bsfLambda}{0}{bsfletters}{'003}
\DeclareMathSymbol{\ssfLambda}{0}{ssfletters}{'003}
\DeclareMathSymbol{\bsfXi}{0}{bsfletters}{'004}
\DeclareMathSymbol{\ssfXi}{0}{ssfletters}{'004}
\DeclareMathSymbol{\bsfPi}{0}{bsfletters}{'005}
\DeclareMathSymbol{\ssfPi}{0}{ssfletters}{'005}
\DeclareMathSymbol{\bsfSigma}{0}{bsfletters}{'006}
\DeclareMathSymbol{\ssfSigma}{0}{ssfletters}{'006}
\DeclareMathSymbol{\bsfUpsilon}{0}{bsfletters}{'007}
\DeclareMathSymbol{\ssfUpsilon}{0}{ssfletters}{'007}
\DeclareMathSymbol{\bsfPhi}{0}{bsfletters}{'010}
\DeclareMathSymbol{\ssfPhi}{0}{ssfletters}{'010}
\DeclareMathSymbol{\bsfPsi}{0}{bsfletters}{'011}
\DeclareMathSymbol{\ssfPsi}{0}{ssfletters}{'011}
\DeclareMathSymbol{\bsfOmega}{0}{bsfletters}{'012}
\DeclareMathSymbol{\ssfOmega}{0}{ssfletters}{'012}

\newcommand{\rvk}{{\mathssf{k}}}	% k

\newcommand{\rvl}{{\mathssf{l}}}	% k
\newcommand{\rvm}{{\mathssf{m}}}	% m

\newcommand{\rvs}{{\mathssf{s}}}	% s

\newcommand{\rvt}{{\mathssf{t}}}	% t

\newcommand{\rvu}{{\mathssf{u}}}	% u

\newcommand{\rvv}{{\mathssf{v}}}	% v

\newcommand{\rvx}{{\mathssf{x}}}	% x, random variable

\newcommand{\rvy}{{\mathssf{y}}}	% y
\newcommand{\rvyr}{{\mathssf{y}}_\mrm{r}}	% y
\newcommand{\rvye}{{\mathssf{y}}_\mrm{e}}	% y

\newcommand{\rvz}{{\mathssf{z}}}	% 
\newcommand{\rvzr}{{\mathssf{z}}_\mathrm{r}}	% z
\newcommand{\rvze}{{\mathssf{z}}_\mathrm{e}}	% z

\renewcommand{\rvk}{\mathsf{\kappa}}

\newcommand{\rvsr}{\mathsf{s}_{r}}
\newcommand{\rvse}{\mathsf{s}_{e}}
\newcommand{\rvst}{\mathsf{s}_{t}}

\newcommand{\orvyr}{\bar{\rvyr}}
\newcommand{\orvye}{\bar{\rvye}}

\usepackage{psfrag}

\begin{document}
\title{Secret-key Agreement with Channel State Information at the Transmitter}

\author{Ashish~Khisti,~\IEEEmembership{Member,~IEEE,}
        Suhas Diggavi,~\IEEEmembership{Member,~IEEE,}
        and~Gregory Wornell,~\IEEEmembership{Fellow,~IEEE}% <-this % stops a space
\thanks{Ashish Khisti is with the Department of Electrical and Computer Engineering, University of Toronto, Toronto, ON, Canada
 e-mail: akhisti@comm.utoronto.ca. Suhas Diggavi is with the Ecole Polytechnique Federale de Lausanne EPFL) and with the University of California, Los Angles (UCLA), USA, email: suhas.diggavi@epfl.ch, while Gregory Wornell is with the Massachusetts Institute of Technology, Cambridge, MIT, USA  email: gww@mit.edu }% <-this % stops a space
\thanks{Parts of this  work were presented at the European Wireless Conference 2010, Lucca, Italy~\cite{khisti:10} and  the IEEE International Symposium on Information Theory (ISIT), Seoul Korea~\cite{khisti:09}.  }% <-this % stops a space
\thanks{The work of Ashish Khisti was supported by a Natural Science and Engineering Research  Council (NSERC) Discovery Grant. This work was also supported by NSF under Grant No. CCF-0515109.}
}

\maketitle

\begin{abstract}
%\boldmath
We study the capacity of secret-key agreement over a wiretap channel with state parameters. The transmitter communicates to the legitimate receiver and the eavesdropper over a discrete memoryless wiretap channel with a memoryless state sequence. The transmitter and the  legitimate receiver generate a shared secret key, that remains secret from the eavesdropper. No public discussion channel is available. The state sequence is known noncausally to the transmitter. We derive lower and upper bounds on the secret-key capacity. The lower bound involves constructing a common state reconstruction sequence at the legitimate terminals and binning the set of reconstruction sequences to obtain the secret-key. For the special case of Gaussian channels with additive interference (\emph{secret-keys from dirty paper channel}) our bounds differ by 0.5 bit/symbol and coincide in the high signal-to-noise-ratio and high interference-to-noise-ratio regimes. For the case when the legitimate receiver is also revealed the state sequence, we establish that our lower bound achieves the the secret-key capacity.   In addition, for this special case, we also propose another scheme that attains the capacity and requires only causal side information at the transmitter and the receiver.
\end{abstract}
%\begin{IEEEkeywords}
%IEEEtran, journal, \LaTeX, paper, template.
%\end{IEEEkeywords}

% \fi
\IEEEpeerreviewmaketitle

\section{Introduction}
Secret keys are a fundamental requirement for any application involving secure communication or computation. An information theoretic approach to secret key generation between two or more terminals was pioneered in~\cite{maurer:93,ahlswedeCsiszar:93} and subsequently extended in~\cite{CsiszarNarayan:04,CsiszarNarayan:10,Gohari:1,Gohari:2}. In the setup considered in these works, the transmitter communicates to a legitimate receiver and the eavesdropper over a memoryless broadcast channel and is interested in generating a secret key shared with the legitimate receiver. The legitimate terminals can also exchange an unlimited number of messages over a public channel.  There has been a significant interest in developing practical approaches for generating shared secret keys between two or more terminals based on such techniques, see e.g.,~\cite{wilsonTseScholtz:07,pap1, pap2, pap3,pap4,pap5, pap6, mathur} and references therein.  

In the present work, we study the secret key agreement capacity over a broadcast channel controlled by a random state variable.  The importance of studying channels with state parameters~\cite{shannon:58,gelfandPinsker:80,wolfowitz:coding} has become increasingly evident in recent times due a variety of   applications including fading channels~\cite{caireShamai99}, broadcast channels~\cite{caireShamai:01} and digital watermaking~\cite{chenWornell:01}. For example in fading channels, the state variable could model the instantaneous  fading coefficient of the channel.  In broadcast channels the state sequence models an  interfering message to another receiver while in watermarking systems the state sequence represents a host sequence on which information message needs to be embedded. Clearly depending on the application the state sequence may be known to either the sender or the receivers or both. In this paper, unless otherwise stated, we assume that the entire state sequence is known to the sender noncausally. As will be discussed,  the seemingly more general case when each receiver also has (a possibly noisy) side information can be easily incorporated in this model.  Some of our results only require causal transmitter side information although we note in advance that we do not consider this problem in detail.

In the present paper we only focus on the case when there is no discussion channel available. We point the reader to our conference papers~\cite{khisti:10,khisti:09} for some results on the case when a public discussion channel is available.  Notice that our setup differs from~\cite{mitrapant:06,chenVinck:06,LiuChen:07} that study the wiretap channel with state parameters and require that the transmitter send a confidential message to the receiver and build on the \emph{wiretap} channel model~\cite{csiszarKorner:78}. Our results indicate that the achievable secret-key rate can be significantly higher compared to the  results in~\cite{mitrapant:06,chenVinck:06,LiuChen:07}. Recently an improved lower bound for the wiretap channel with causal state information at the transmitter and receiver has been reported in~\cite{khiangElGamal:10}. Interestingly it uses a block markov coding scheme, where a secret key is generated in each block as an intermediate step. 

After the conference papers~\cite{khisti:09,khisti:10} on which this paper is based appeared, the authors became aware about a recent work~\cite{prabhakaran:09} where a similar secret-key agreement scheme over channels with noncausal channel state information is presented. This scheme is used in constructing a coding scheme that provides a tradeoff between secret-key and secret-message transmission. The paper~\cite{prabhakaran:09} however does not fully explore the problem of secret key agreement over wiretap channels with state parameters.  In particular to the best of our knowledge, it does not have the results in the present paper such as an upper bound on the secret-key capacity, the asymptotic optimality of the lower bound for the Gaussian case or the secret-key capacity for the case of symmetric CSI. 

\section{Problem Statement}

\subsection{Channel Model}
\label{subsec:Model}
The channel model has three terminals --- a sender, a receiver and an eavesdropper.
The sender communicates with the other two terminals over a
discrete-memoryless-channel  controlled by a random state parameter. The transition probability of the channel is
$p_{\rvyr, \rvye|\rvx,\rvs}(\cdot)$ where $\rvx$ denotes the channel input
symbol, whereas $\rvyr$ and $\rvye$ denote the channel output
symbols at the receiver and the eavesdropper respectively. The
symbol $\rvs$ denotes a state variable that controls the channel
transition probability. We assume that it is  independent and identically distributed (i.i.d.) from a
distribution $p_{\rvs}(\cdot)$ in each channel use. Further, the entire
sequence $\rvs^n$ is known to the sender before the communication
begins. 

As explained in section~\ref{subsec:ExtendedModel} the model generalizes easily to take into account correlated side information sequence at each of the receivers. 

\begin{figure}
\centering \psfrag{sr}{$\rvs$}\psfrag{k}{$\rvk$}
\includegraphics[scale=0.35]{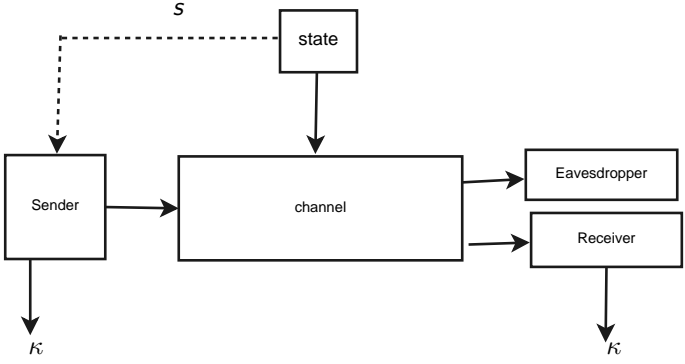}
\caption{Wiretap channel controlled by a state parameter. The
channel transition probability $p_{y_r,y_e|x,s}$ is controlled by
a state parameter $\rvs$. The entire source sequence $\rvs^n$ is
known to the sender but not to the receiver or the eavesdropper. The
sender and receiver generate a secret key $\rvk$ at the end of the
transmission.}\label{fig:wiretap}
\end{figure} 

\subsection{Secret-Key Capacity}
\label{subsec:SecKeyCap}
A length $n$ encoder is defined as follows. The sender samples a
random variables $\rvm_\rvx$ from the conditional distribution
$p_{\rvm_\rvx|\rvs^n}(\cdot|s^n)$. The encoding function produces a
channel input sequence 
\begin{equation}\rvx^n = f_n(\rvm_\rvx, \rvs^n) \label{eq:def:Enc}\end{equation} and transmits it
over $n$ uses of the channel. At time $i$ the symbol $\rvx_i$ is
transmitted and the legitimate receiver and the eavesdropper observe
output symbols $\rvy_{ri}$ and $\rvy_{ei}$ respectively, sampled
from the conditional distribution
$p_{\rvyr,\rvye|\rvx,\rvs}(\cdot)$. The sender and receiver compute
secret keys \begin{equation} \rvk = g_n(\rvm_\rvx,\rvs^n), \qquad \rvl = h_n(\rvyr^n).\label{eq:def:Dec}\end{equation} A
rate $R$ is achievable if there exists a sequence of encoding
functions such that for some sequence $\eps_n$ that vanishes as
$n\rightarrow\infty$, we have that $\Pr(\rvk\neq\rvl) \le \eps_n$
and \begin{equation}\label{eq:rate}\frac{1}{n}H(\rvk) \ge R - \eps_n,\end{equation} and \begin{equation}
\label{eq:secNoInterac} \frac{1}{n}I(\rvk;\rvye^n) \le \eps_n.
\end{equation} The largest achievable rate is the secret-key capacity.

\subsection{Extended Model}
\label{subsec:ExtendedModel}
In our proposed model we are assuming the state variable is only known to the transmitter and not to the receiving terminals. A more general model involves a state variable that can be decomposed into $\rvs = (\rvst, \rvsr, \rvse, \rvs_0)$ where the sequence $\rvst^n$ is revealed noncausally to the sender whereas $\rvsr^n$ and $\rvse^n$ are revealed to the legitimate receiver and the eavesdropper respectively while $\rvs_0^n$  is not revealed to any of the terminals. It turns out that the model in section~\ref{subsec:Model} includes this extended model. The secret-key capacity for this new model is identical to the secret-key capacity of a particular model in section~\ref{subsec:Model} defined by: $\orvyr = (\rvyr, \rvsr)$ and $\orvye = (\rvye, \rvse)$ and the channel transition probability
\begin{equation}
p(\bar{y}_r, \bar{y}_e|s_t,x) = \sum_{s_0} p(y_r, y_e|s_0,s_r,s_e,s_t,x)p(s_0,s_r,s_e|s_t).\label{eq:equiv}
\end{equation}
The equivalence can be established by noting that the modified channel preserves the same knowledge of the side information sequences as the original problem, the rate and equivocation terms only depend on the joint distribution $p(\bar{y}_r^n, \bar{y}_e^n,x^n,s_t^n)$ and for any input distribution $p(x^n|s_t^n)$, the extended channel satisfies
\begin{equation}
p(\bar{y}_r^n, \bar{y}_e^n| x^n,s_t^n) = \prod_{i=1}^n p(\bar{y}_{ri}, \bar{y}_{ei}|x_i, s_{ti}), 
\label{eq:jointPDF}\end{equation}
where each term on the right hand side of~\eqref{eq:jointPDF} obeys~\eqref{eq:equiv}.

We omit a detailed proof in interest of space and point to the reader to~\cite[pp.~17---25]{sideInfo}\cite[Chapter 7, pp.~7-54]{ElGamalKim} for an analogous observation. Note that our model inherently uses the asymmetry in channel state knowledge between the eavesdropper and the legitimate receiver for secret key generation. While as discussed in this subsection, it can be easily extended to incorporate receiver side information, for simplicity in exposition we will suppress the availability of side information at the receivers. 

\section{Main Results}
\label{sec:MainResults}
We summarize the main results of this paper in this section.

\subsection{Capacity Bounds}
\label{subsec:CapacityBounds}
We first provide an  achievable rate (lower bound) on the secret-key
capacity.
\begin{thm}
An achievable secret-key rate is
\begin{equation}
\label{eq:lowerBoundNoDisc} R^- = \max_{p_\rvu, p_{\rvx|\rvs,\rvu}}
I(\rvu;\rvyr) - I(\rvu;\rvye),
\end{equation}
where the maximization is over all auxiliary random variables $\rvu$
that satisfy the Markov condition $\rvu \rightarrow (\rvx, \rvs)
\rightarrow (\rvyr,\rvye)$ and furthermore satisfy the constraint
that
\begin{equation}
I(\rvu;\rvyr)-I(\rvu;\rvs) \ge 0. \label{eq:lowerBoundNoDiscCons}
\end{equation}\label{thm:lbNoDisc}
\end{thm}

The intuition behind the coding scheme is as follows. Upon observing
$\rvs^n$, the sender communicates the best possible reproduction
$\rvu^n$ of the state sequence to the receiver Now both the sender
and the receiver observe a common sequence $\rvu^n$. The set of all
codewords $\rvu^n$ is binned into $2^{nR^-}$ bins and the bin-index
is declared to be the secret key. Note that the problem of communicating 
a state sequence with common knowledge to the receiver is studied in~\cite{steinberg:08,steinberg:09}.
This setup requires that the reconstruction sequence satisfy a certain distortion measure with respect to the
state sequence. In contrast the  common reconstruction sequence in this problem is an intermediate step
used to generate a common secret key.

While we do not have a matching upper bound to Theorem~\ref{thm:lbNoDisc} the following result provides an upper bound to the secret-key capacity that is amenable to numerical evaluation. 

\begin{thm}
The secret-key capacity  is upper
bounded by $C \le R^+$, where
\begin{equation}
\label{eq:upperBoundNoDisc} R^+ = \min_{p_{\rvyr,\rvye|\rvx,\rvs} \in
\cP}\max_{p_{\rvx|\rvs}} I(\rvx,\rvs;\rvyr|\rvye),
\end{equation}
where $\cP$ denotes all the joint distributions
$p^\star_{\rvyr,\rvye|\rvx,\rvs}$ that have the same marginal
distribution as the original channel. \label{thm:ubNoDisc}
\end{thm}

The intuition behind the upper bound is as follows. We create a
degraded channel by revealing the output of the eavesdropper to the
legitimate receiver. We further assume a channel with two inputs
$(\rvx^n,\rvs^n)$ i.e., the state sequence $\rvs^n$ is not
arbitrary, but rather a part of the input codeword with distribution
$p_{\rvs}$. The secrecy capacity of the resulting wiretap channel
is then given by $I(\rvx,\rvs;\rvyr|\rvye)$.

Note that the problem of secret-key agreement is
different from the secret-message transmission problem considered
in~\cite{mitrapant:06,chenVinck:06,LiuChen:07}. This is because the
secret-key can be an arbitrary function of the state sequence (known
only to the transmitter) whereas the secret-message needs to be
independent function of the state sequence. For comparison, the best
known lower bound on the secret-message transmission problem is
stated below.
\begin{prop}\cite{mitrapant:06,chenVinck:06,LiuChen:07}
An achievable secret message rate for the wiretap channel with
noncausal transmiter channel state information (CSI) is
\begin{equation}
R =\max_{p_{\rvu},p_{\rvx|\rvu,\rvs}}I(\rvu;\rvyr) -
\max\left(I(\rvu;\rvs), I(\rvu;\rvye)\right).\label{eq:secMsgRate}
\end{equation}
\end{prop}
We note that the secret-key rate~\eqref{eq:lowerBoundNoDisc} is in general
strictly better than the secret-message rate~\eqref{eq:secMsgRate}.

\subsection{Secret Keys from Dirty Paper Coding}
\label{subsec:Gaussian} We study the Gaussian case under
an average power constraint. The channel to the legitimate receiver
and the eavesdropper is expressed as:
\begin{equation}\begin{aligned}
\rvyr &= \rvx + \rvs + \rvzr\\
\rvye &= \rvx + \rvs + \rvze,
\end{aligned}\label{eq:GaussianModel}\end{equation} where
$\rvzr \sim \cN(0,1)$ and $\rvze \sim \cN(0,1+\Delta)$ denote the
additive white Gaussian nose and are assumed to be sampled
independently. The state parameter $\rvs \sim \cN(0,Q)$ is also
sampled i.i.d.\ at each time instance and is independent of both
$\rvzr$ and $\rvze$. Furthermore, the channel input satisfies an
average power constraint $E[\rvx^2] \le P$. We assume $\rvs^n$ to be noncausally known to the sender
but not to any other terminals. 

Thus the parameter $P$ denotes the signal-to-noise ratio, the
parameter $Q$ denotes the interference-to-noise-ratio, whereas
$\Delta$ denotes the degradation level of the eavesdropper. We now
provide lower and upper bounds on the secret-key capacity\footnote{Interestingly in the presence of public discussion, we have been able to characterize the secret-key capacity~\cite{khisti:10}.}. We limit our analysis to the case when $P \ge 1$.

\begin{figure*}
\begin{minipage}[b]{0.5\linewidth}
\centering
\includegraphics[width=10cm]{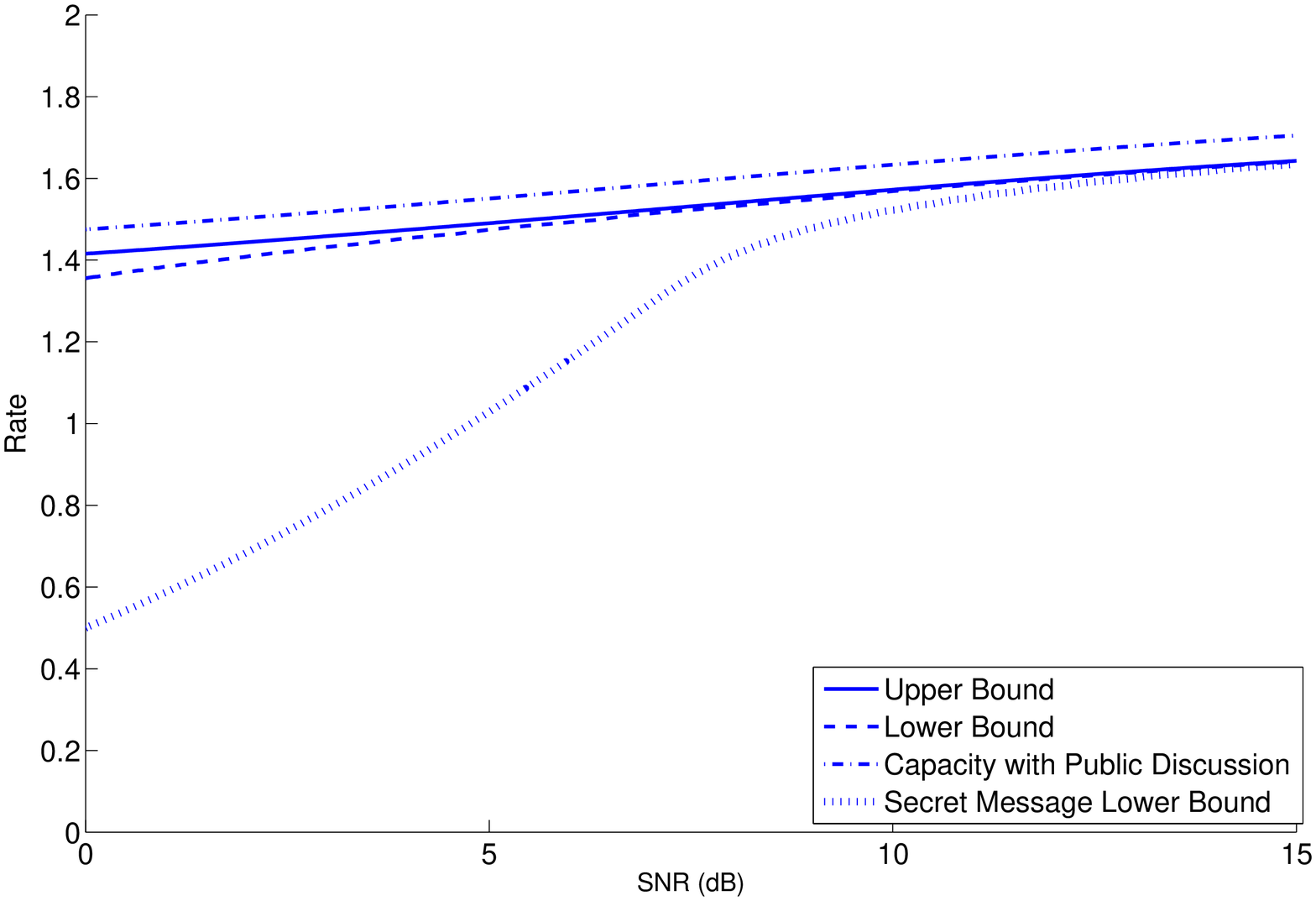}
%\caption{Rates as a function of SNR} \label{fig:snrfig}
\end{minipage}
\hspace{0.5cm} % To get a little bit of space between the figures
\begin{minipage}[b]{0.5\linewidth}
\centering
\includegraphics[width=10cm]{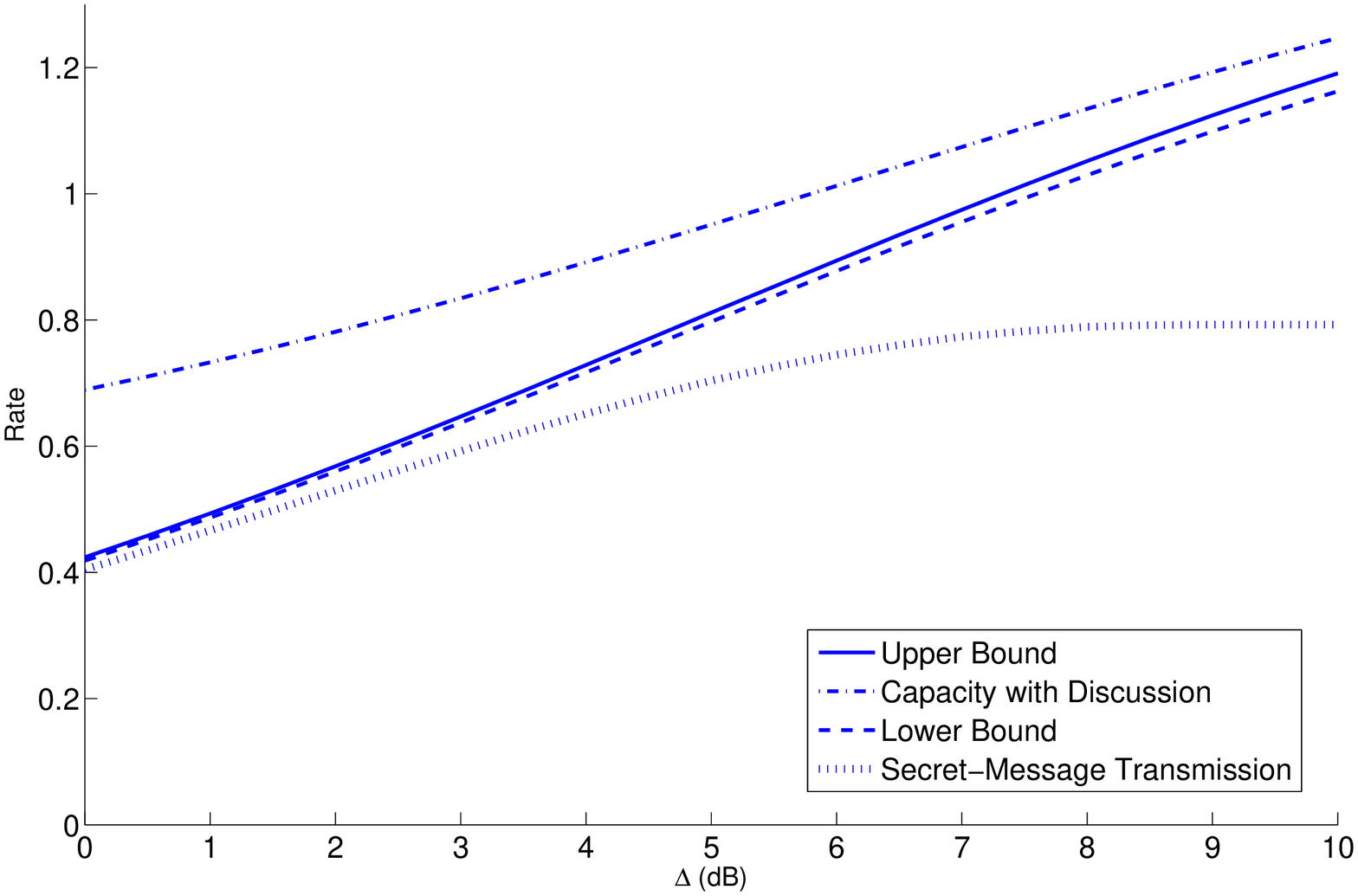}
%\caption{Rates as a function of $\Delta$}\label{fig:deltafig}
\end{minipage}
\caption{Bounds on the capacity of the ``secret-keys from dirty paper" channel.
 In the left figure, we
plot the bounds on capacity as a function of SNR (dB) when $Q = 10$ dB and
$\Delta = 10$ dB. The upper-most curve is the capacity with
public-discussion~\cite{khisti:10} whereas the next two curves denote the upper and
lower bounds on the capacity as stated in Prop.~\ref{prop:ubGaussNoDisc} and Prop.~\ref{prop:lbGaussNoDisc}. The dotted curve is the secret message transmission lower bound~\eqref{eq:secMsgRate} evaluated for a jointly Gaussian input distribution. In
the right figure we vary the degradation level at the eavesdropper $\Delta$ (in dB) and compute the secret-key rates for $P=2$ 
and $Q = 2$. The upper-most curve is the secret-key capacity with public discussion~\cite{khisti:10}, the next two curves are the upper and the lower bounds, whereas the dotted curve is the secret message transmission rate evaluated for Gaussian inputs. }
\label{fig:GaussNumerical}
\end{figure*}

\begin{prop}
Assuming that $P \ge 1$, a lower bound on the secret-key agreement capacity is
capacity is given by,
\begin{multline}
R^- =  \frac{1}{2}\log\left(1 + \frac{\Delta(P+Q+2\rho \sqrt{PQ})}{P+Q+1+\Delta + 2\rho \sqrt{PQ}}\right),
\label{eq:lbGaussNoDisc}\end{multline}
where $|\rho| < 1$ and
\begin{equation}
P(1-\rho^2) =1 - \frac{1}{P+Q+1}.\label{eq:rhoCons}
\end{equation}
\label{prop:lbGaussNoDisc}
\end{prop}
\begin{prop}
An upper bound on the secret-key
capacity is given by,
\begin{equation}
R^+ = \frac{1}{2}\log\left(1 +
\frac{\Delta(P+Q + 2\sqrt{PQ})}{P+Q+1+\Delta + 2\sqrt{PQ}}\right)\label{eq:ubGaussNoDisc}
\end{equation}\label{prop:ubGaussNoDisc}
\end{prop}
It can be readily verified that the upper and lower bounds are close in several interesting regimes. In Fig.~\ref{fig:GaussNumerical} we numerically plot these bounds and state some properties  below. We omit the proof due to space constraints. 
\begin{prop}
The upper and lower bounds on secret-capacity  satisfy the following
\begin{align}
&R_+ - R_- \le \frac{1}{2}~~\mrm{bit/symbol} \label{eq:univGap} \\
&\lim_{P\rightarrow\infty} R_+ - R_- = 0 \label{eq:highINR}\\
&\lim_{Q\rightarrow\infty} R_+ - R_- = 0 \label{eq:lowINR}
\end{align}
\end{prop}

\subsection{Symmetric CSI}
\label{subsec:CSI}
Consider the special case where the state sequence $\rvs$ is also revealed to the legitimate receiver. In this case we have a complete characterization of the secret-key capacity.
\begin{thm}
The secret-key capacity for the channel model in section~\ref{subsec:Model} when the state sequence $\rvs^n$ is also revealed to the decoder is given by
\begin{equation}
C_\mrm{sym} = \max_{p_{\rvu|\rvs(\cdot)}p_{\rvx|\rvu,\rvs(\cdot)}} I(\rvu;\rvyr|\rvs) - I(\rvu;\rvye|\rvs) + H(\rvs|\rvye),\label{eq:symCap}
\end{equation}
where the maximization is over all auxilary random variables $\rvu$ that obey the Markov chain
$\rvu \rightarrow (\rvx,\rvs) \rightarrow (\rvyr,\rvye)$. Additionally it suffices to limit the cardinality of the auxiliary variable to $|\cS|(1 + |\cX|)$ in~\eqref{eq:symCap}.
\label{thm:symCSI}
\end{thm}

The achievability in~\eqref{eq:symCap} follows from~\eqref{eq:lowerBoundNoDisc} by augmenting $\orvyr = (\rvyr, \rvs)$. Observe that~\eqref{eq:lowerBoundNoDiscCons} is redundant as 
$I(\rvu;\rvyr,\rvs)-I(\rvu;\rvs) \ge 0$ holds. Furthermore the expression in~\eqref{eq:lowerBoundNoDisc} can be simplified as follows \begin{align}
R^- &= \max_{p_\rvu, p_{\rvx|\rvs,\rvu}}
I(\rvu;\rvyr,\rvs) - I(\rvu;\rvye)  \notag\\
&= \max_{p_\rvu, p_{\rvx|\rvs,\rvu}}
I(\rvu;\rvyr |\rvs) - I(\rvu;\rvye | \rvs)  + I(\rvs ; \rvu |\rvye) \label{eq:ft1}\\
&= \max_{p_\rvu, p_{\rvx|\rvs,\rvu}} 
I(\rvu;\rvyr |\rvs) - I(\rvu;\rvye | \rvs)  + H(\rvs |\rvye) \label{eq:ft2}
\end{align}
where the last relation follows by noting that if $\rvu$ is an optimal choice in~\eqref{eq:ft1} then by selecting $\rvu^\star = (\rvu,\rvs)$ will leave the difference in the two mutual information terms unchanged but increase the second term $H(\rvs|\rvye)$ as specified in~\eqref{eq:ft2}. Notice that~\eqref{eq:ft2} is identical to~\eqref{eq:symCap}. The converse follows by an application of Csiszar's Lemma and is provided in section~\ref{subsec:converse}

We provide another achievability scheme for Theorem~\ref{thm:symCSI} that only requires causal knowledge of $\rvs^n$  at the encoder.  The scheme is based on the following interpretation of~\eqref{eq:symCap}. The term $I(\rvu;\rvyr|\rvs) - I(\rvu;\rvye|\rvs)$ is the rate of a multiplexed wiretap codebook constructed assuming that all the three terminals have knowledge of $\rvs^n$. The second term $H(\rvs|\rvye)$ is the rate of the additional secret key that can be produced by exploiting the fact that $\rvs^n$ is only known to the sender and the legitimate terminal. This scheme is causal since the multiplexed code uses only current state to decide which codebook to use. Furthermore, since the state is known to the sender and receiver, the second term is also causal. 

We note that the capacity expression~\eqref{eq:symCap} captures an interesting tension between two competing forces in choosing the optimal distribution. To maximize the contribution of the rate obtained from the multiplexed wiretap codebook, it is desirable to select $\rvu$ to be strongly correlated with $\rvs$. However doing so will leak more information about $\rvs$ to the wiretapper and reduce the rate contribution of the second codebook. To maximize the contribution of the common state sequence, we need to select an input that masks the state sequence from the eavesdropper~\cite{merhavShamai}.  We illustrate this tradeoff  via an example in section~\ref{subsec:NumericalExample}. 

Finally it can be easily verified that the the expression~\eqref{eq:symCap} simplifies in the following special case.
\begin{corol}
\label{corol:LessNoisySymCap}
Suppose that for each $s \in \cS$ the channel $p_{\rvyr,\rvye|\rvs=s, \rvx}(y_r,y_e|s,x)$ is such that the eavesdropper's channel is less noisy compared to the legitimate receiver's channel. Then the secret-key capacity with $\rvs^n$ revealed to both the legitimate terminals is
\begin{equation}
C = \max_{p_{\rvx|\rvs}} H(\rvs|\rvye).\label{eq:LessNoisySymCap}
\end{equation}
\end{corol}
Intuitively, when the wiretap channel cannot contribute to the secrecy,~\eqref{eq:LessNoisySymCap} states that transmitter should select an input that masks the state from the output as much as possible.
\subsection{Symmetric CSI: Numerical Example}
\label{subsec:NumericalExample}
\begin{figure*}
\begin{minipage}[b]{0.5\linewidth} % A minipage that covers half the page
\centering
\includegraphics[scale=0.3]{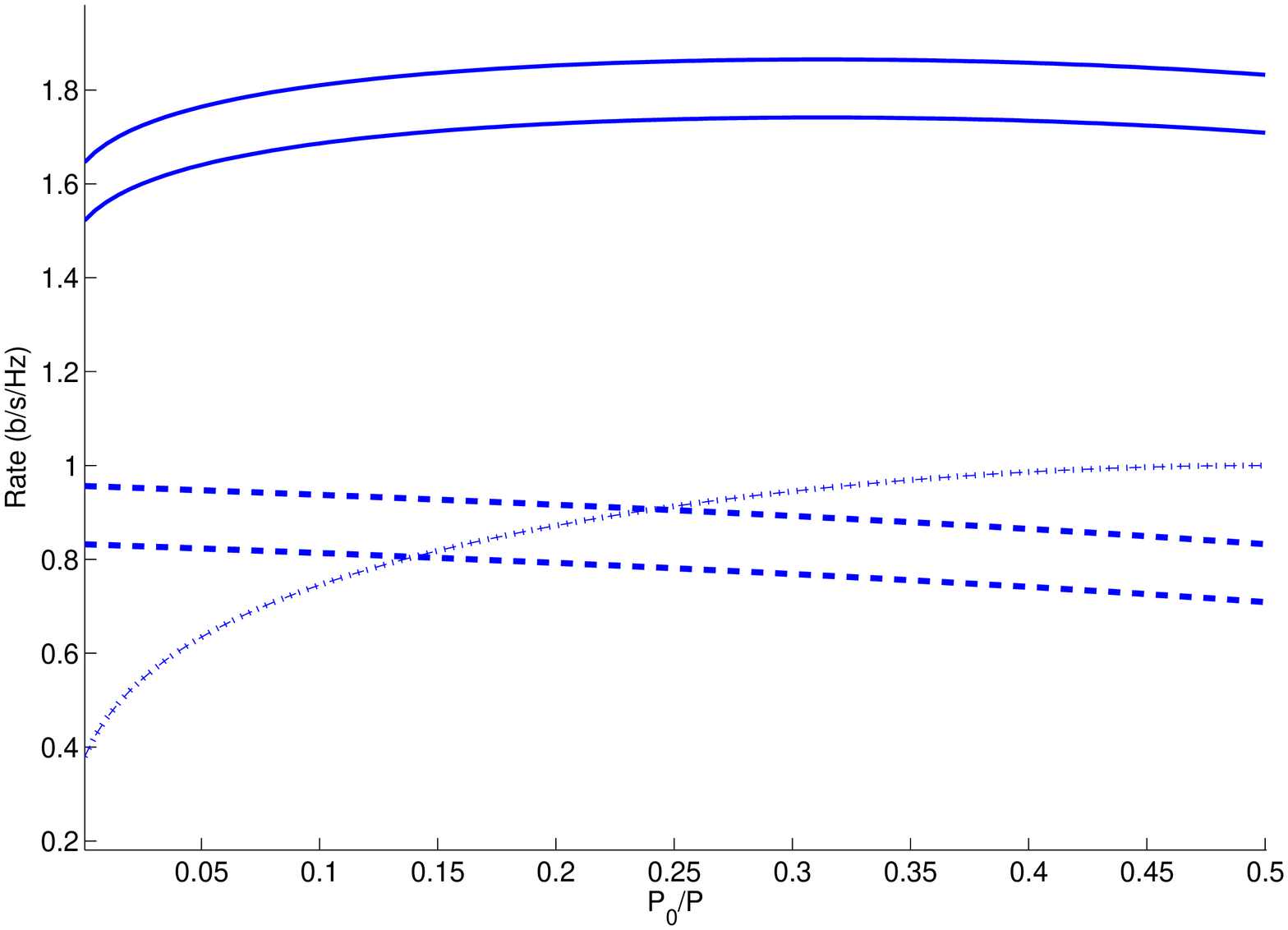}
\caption{The achievable secret-key rate as a fraction of power allocated to the state $\rvsr=0$ and SNR = 17 dB. The solid curve denotes the secret-key rate, the dashed curve denotes the rate of the secret-message, while the dotted curve denotes the conditional entropy term $H(\rvsr|\rvse = 1, \rvye=y_e)$ in~\eqref{eq:numerical1}. The upper solid and dashed curves  denote the case of public discussion while the  other solid and dashed curves denote the case of no public discussion.}\label{fig:trade}
\end{minipage}\hspace{0.5cm}
\begin{minipage}[b]{0.5\linewidth}
\centering
\includegraphics[scale = 0.3]{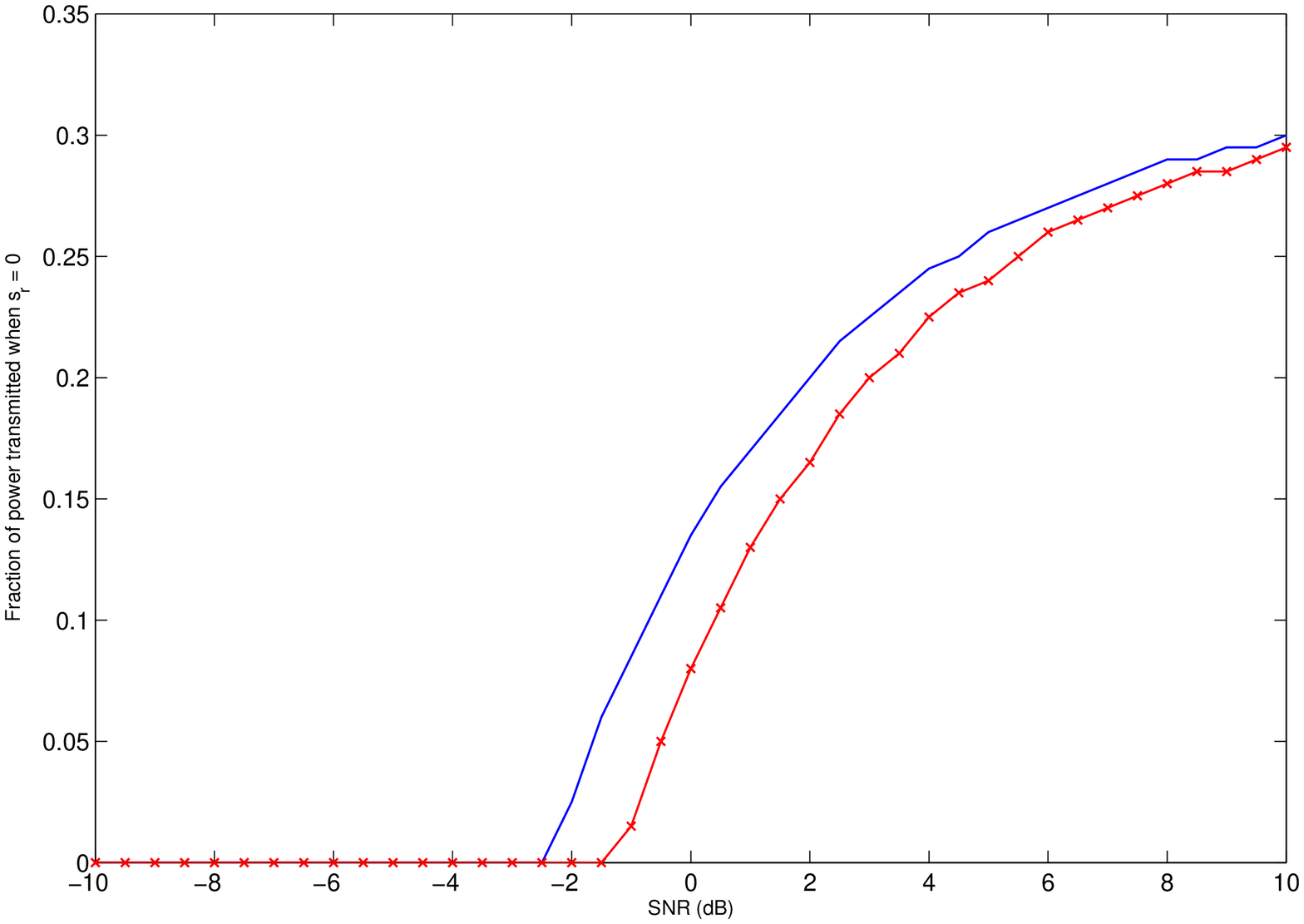}
\caption{Optimal fraction of power that must be allocated to the
state $\rvsr =0$ to maximize the secret-key rate with Gaussian
inputs.  The curve marked with a ($\times$) denotes the case of
public discussion while the other curve denotes the case of  no
public discussion.}\label{fig:opt}
\end{minipage}
\end{figure*}

It can be easily seen that for the dirty paper coding example in section~\ref{subsec:Gaussian}, the secret-key capacity when $\rvs$ is also revealed to the legitimate receiver is infinity. More generally higher the entropy of $\rvs$, higher will be the gains in the secret-key capacity with symmetric CSI.  In this section illustrate the secret-key rate for an on-off channel for the receivers:
\begin{equation}
\begin{aligned}
\rvyr&= \rvs_{r} \rvx + \rvzr\\
\rvye&= \rvs_{e} \rvx + \rvze,
\end{aligned}\label{eq:fading}\end{equation}
where both $\rvs_{r}, \rvs_{e} \in \{0,1\}$, the random variables
are mutually independent and $\Pr(\rvs_{r}=0)= \Pr(\rvse = 0) =
0.5$. Furthermore we assume that $\rvsr$ is revealed to the
legitimate terminals, whereas the eavesdropper is revealed
$\tilde{\rvye}=(\rvse,\rvye)$. The noise random variables are
mutually independent, zero mean and unit variance Gaussian random
variables and the power constraint is that $E[\rvx^2] \le P$.

We evaluate the secret-key rate expression for Gaussian inputs i.e., $\rvu=\rvx \sim \cN(0,  P_0)$
when $\rvs_\mrm{r} = 0$ and $\rvu=\rvx \sim \cN(0, P_1)$ when $\rvs_\mrm{r} = 1$. Further to satisfy the average power constraint we have that
$P_0 +P_1 \le 2P$. An achievable rate from Theorem~\ref{thm:symCSI} 
\begin{align}
R &= I(\rvx;\rvyr|\rvsr) - I(\rvx;\tilde{\rvye}| \rvsr) + H(\rvsr| \tilde{\rvye}) \\
&=I(\rvx;\rvyr|\rvsr) - I(\rvx;\rvye,\rvse| \rvsr) + H(\rvsr| \rvse, \rvye)\\
&= \frac{1}{8}\log(1 + P_1) + \frac{1}{2}E_{\rvye}[H(p(y_e),1-p(y_e))] + \frac{1}{2},  \label{eq:numerical1}
\end{align}
where we have introduced
\begin{equation}
p(y_e) = \frac{\cN_{y_e}(0,P_0+1)}{\cN_{y_e}(0,P_0+1) +
\cN_{y_e}(0,P_1+1)}
\end{equation}
the aposterior distribution $\Pr(\rvsr=0|y_e)$ and the notation $\cN_{\rvye}(0,\sigma^2)$ denotes the zero mean Gaussian distribution with variance $\sigma^2$ evaluated at $\rvye$ and where~\eqref{eq:numerical1} follows through a straightforward computation.  

\iffalse
&= \Pr(\rvsr = 1, \rvse = 0) I(\rvx; \rvyr|\rvsr = 1) + \notag\\
&\quad\quad \Pr(\rvse = 1)H(\rvsr|\rvse = 1, \rvye) + \Pr(\rvse = 0)H(\rvsr)\\
&= \frac{1}{8}\log(1 + P_1) + \frac{1}{2}H(\rvsr|\rvse = 1, \rvye) + \frac{1}{2}\\
&= \frac{1}{8}\log(1 + P_1) + \frac{1}{2}E_{\rvye}[H(\rvsr|\rvse = 1, \rvye=y_e)] + \frac{1}{2} \\
\fi

In Fig.~\ref{fig:trade} we numerically evaluate this rate for $\mrm{SNR} = 17$ dB. For comparison we also plot the corresponding rate with public discussion~\cite{khisti:09} 
\begin{equation}R_\mrm{disc}= \frac{1}{8}\log(1 + 2P_1) + \frac{1}{2}E_{\rvye}[H(p(y_e),1-p(y_e))] + \frac{1}{2}.\label{eq:numerical2}\end{equation}

In Fig.~\ref{fig:trade} the solid curves show the secret key rate with and without public discussion,  while the dashed curve is the entropy
$H(\rvsr|\rvse=1,\rvye)$ and the dotted curve denotes contribution of the wiretap code. Note that in general there is a tradeoff
between these two terms. To maximize the  conditional entropy we set $P_0 = P_1=P/2$, while to maximize the wiretap codebook rate we need to set
$P_0 = 0$ and $P_1 = P$. The resulting secret-key rate is maximized by selecting a power allocation that balances these two terms.  The optimum fraction of power  transmitted in the state $\rvsr=0$ as a function of the signal to noise ratio is shown in
Fig.~\ref{fig:opt}. Note that no power is transmitted when the
signal-to-noise ratio is below $\approx -2.5 dB$. In this regime the
channels are sufficiently noisy so that $H(\rvsr|\rvye,\rvse=1)
\approx 1$ even with $P_0 = 0$ and hence all the available power is
used for transmitting the secret-message. As the signal-to-noise
ratio increases more information regarding $\rvsr$ gets leaked to
the eavesdropper and to compensate for this effect, a non-zero
fraction of  power is transmitted when $\rvsr = 0$.

\section{Secret key generation with noncausal Transmitter CSI }

In this section we provide Proofs of Theorem~\ref{thm:lbNoDisc} and~\ref{thm:ubNoDisc} i.e., the coding scheme and the upper bound for the secret key agreement problem. 

\subsection{Proof of Theorem~\ref{thm:lbNoDisc}}
The coding theorem involves constructing a common sequence $\rvu^n$ at the legitimate terminals and using it to generate a secret key.
\begin{figure}
	\centering
	\includegraphics[scale=0.35]{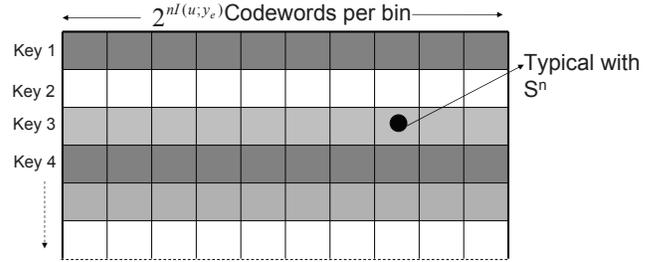}
\caption{Codebook for the secret key agreement problem. A total of $\approx 2^{nI(\rvu;\rvyr)}$ codewords are generated i.i.d. $p_\rvu(\cdot)$ and partitions into $2^{nR}$ bins so that thare are $2^{nI(\rvu;\rvye)}$
sequences in each bin. Given $\rvs^n$, a jointly typical sequence $\rvu^n$ is selected and its bin index constitutes the secret key. }
\end{figure}

\subsubsection{Codebook Generation}
Assume that the input distribution is such that $I(\rvu;\rvyr) > I(\rvu;\rvs)$ as required in Theorem~\ref{thm:lbNoDisc}. Let $\eps_n$ be a sequence of non-negative numbers that goes to zero such that $2\eps_n < I(\rvu;\rvyr) - I(\rvu;\rvs)$.
\begin{itemize}
\item Generate a total of  $T= 2^{n(I(\rvu;\rvyr) -
2\eps_n)}$ sequences. Each sequence is sampled i.i.d.\ from a
distribution $p_\rvu(\cdot)$. Label them $\rvu_1^n,\ldots, \rvu_T^n$.

\item Select a rate $R = I(\rvu;\rvyr)- I(\rvu;\rvye)-\eps_n$ and
randomly partition the set sequences in the previous step into
$2^{nR}$ bins so that there are $2^{n(I(\rvu;\rvye)-\eps_n)}$
sequences in each bin.
\end{itemize}

\subsubsection{Encoding}
\begin{itemize}
\item Given a state sequence $\rvs^n$ the encoder selects a sequence
$\rvu^n$ randomly from the list of all possible sequences that are
jointly typical with $\rvs^n$. Let the index of this sequence be $L$.
\item At time $i=1,2,\ldots, n$ the encoder transmits symbol
$\rvx_i$ generated by sampling the distribution
$p_{\rvx|\rvu,\rvs}(\cdot| u_i, s_i)$.
\end{itemize}

\subsubsection{Secret-key generation}
\begin{itemize}
\item The decoder upon observing $\rvyr^n$  finds a sequence
$\rvu^n$ jointly typical with $\rvyr^n$.
\item Both encoder and the decoder declare the bin-index of $\rvu^n$
to be the secret-key.
\end{itemize}

\subsubsection{Error Probability Analysis}
An error occurs only if one of the following events occur:
\begin{align}
\cE_1 &= \{(\rvu^n(l), \rvs^n) \notin \cT_\eps^n(\rvu,\rvs) \text{ for all } 1 \le l \le T \}\\
\cE_2 &= \{(\rvu^n(L), \rvyr^n) \notin \cT_\eps^n(\rvu,\rvyr)  \}\\
\cE_3 &= \{(\rvu^n(l), \rvyr^n) \in \cT_\eps^n(\rvu,\rvyr) \text{ for some }  l \neq L \}
\end{align}
Since the number of sequences $T > 2^{nI(\rvu;\rvs)}$ it follows from the Covering Lemma~\cite[Chapter 3]{ElGamalKim} that $\Pr(\cE_1) \rightarrow 0$ as $n\rightarrow \infty$. Furthermore let $\cE_1^c = \{(\rvu^n,\rvs^n,\rvx^n) \in \cT_{\eps'}^n(\rvu,\rvs,\rvx)\}$ and $\Pr(\cE_1^c) \rightarrow 1$ as $n \rightarrow \infty$ for any
$\eps' < \eps$. Since $p(\rvyr^n|\rvu^n(L), \rvx^n,\rvs^n) = \prod_{i=1}^n p(\rvy_{ri}|\rvu_i,\rvx_i,\rvs_i)$ it follows from the conditional typicality Lemma~\cite[Chapter 2]{ElGamalKim} that
$\Pr(\cE_2 \cap \cE_1^c) \rightarrow 0$ as $n\rightarrow\infty$. Finally since every $\rvu^n(l)$ is generated i.i.d.\ $p_{\rvu}(u_i)$ and is independent of $\rvyr^n$ 
for $l \neq L$ it follows from the Packing Lemma~\cite[Chapter 3]{ElGamalKim} that $\Pr(\cE_3) \rightarrow 0$ if $T < 2^{nI(\rvu;\rvyr)}$.
  
\subsubsection{Secrecy Analysis}
We need to show that for the proposed encoder and decoder, the
equivocation at the eavesdropper satisfies
\begin{equation}
\label{eq:equivCondn}\frac{1}{n}H(\rvk|\rvye^n)= I(\rvu;\rvyr)-
I(\rvu;\rvye) + o_n(1),
\end{equation}where $o_n(1)$ is a term that goes to zero as $n\rightarrow
\infty$.

Note that while the key $\rvk$ in general can be a function of $(\rvs^n,\rvm_\rvx)$ as indicated in~\eqref{eq:def:Enc}, in our coding scheme the secret key is a deterministic functino of $\rvu^n$ and hence we have

 \begin{align*} \frac{1}{n}H(\rvk|\rvye^n)&=
\frac{1}{n}H(\rvk,\rvu^n|\rvye^n) - \frac{1}{n}H(\rvu^n|\rvye^n,
\rvk)\\
&= \frac{1}{n}H(\rvu^n|\rvye^n) - \frac{1}{n}H(\rvu^n|\rvye^n,
\rvk)\\
&= \frac{1}{n}H(\rvu^n|\rvye^n) - \eps_n\\
\end{align*}
where the last step follows from the fact that there are 
$T_0 = 2^{n(I(\rvu;\rvye)-\eps_n)}$ sequences in each bin. Again applying the packing lemma
we can show that with high probability the eavesdropper uniquely finds the codeword $\rvu^n(L)$ jointly typical with $\rvye^n$ in this set
and hence Fano's Inequality implies that
$$\frac{1}{n}H(\rvu^n|\rvye^n,
\rvk) \le \eps_n.$$
It remains to show that
$$\frac{1}{n}H(\rvu^n|\rvye^n) \ge I(\rvu;\rvyr)-
I(\rvu;\rvye) - o_n(1).$$ Using the chain rule of the joint entropy we have
\begin{align}
&\frac{1}{n}H(\rvu^n|\rvye^n) = \frac{1}{n}H(\rvu^n) +
\frac{1}{n}H(\rvye^n|\rvu^n)- \frac{1}{n}H(\rvye^n)\\
&= \frac{1}{n}H(\rvu^n) + \frac{1}{n}H(\rvye^n|\rvu^n,\rvs^n)-
\frac{1}{n}H(\rvye^n) + \frac{1}{n}I(\rvs^n;\rvye^n|\rvu^n).
\label{eq:condEntLB}
\end{align}
We now appropriately bound each term in~\eqref{eq:condEntLB}. First
note that since the sequence $\rvu^n$ is uniformly distributed among
the set of all possible codeword sequences, it follows that
\begin{align}
\frac{1}{n}H(\rvu^n)&= \frac{1}{n}\log_2 |\cC|  \notag\\
&= I(\rvu;\rvyr) - 2\eps_n \label{eq:EquivTerm1}
\end{align}

Next,   as verified below, the channel to the
eavesdropper $(\rvu^n,\rvs^n) \rightarrow \rvye^n$, is memoryless:
\begin{align*}
&p_{\rvye^n|\rvu^n,\rvs^n}(y_e^n|u^n,s^n) \\&= \sum_{x^n \in \cX^n}
p_{\rvye^n|\rvu^n,\rvs^n,\rvx^n}(y_e^n|u^n,s^n,x^n)p_{\rvx^n|\rvu^n,\rvs^n}(x^n|u^n,s^n)\\
&= \sum_{x^n \in \cX^n} \prod_{i=1}^n
p_{\rvy_{e}|\rvu,\rvs,\rvx}(y_{e,i}|u_i,s_i,x_i)p_{\rvx|\rvu,\rvs}(x_i|u_i,s_i)\\
&= \prod_{i=1}^n
\sum_{x_i \in \cX}p_{\rvy_{e}|\rvu,\rvs,\rvx}(y_{e,i}|u_i,s_i,x_i)p_{\rvx|\rvu,\rvs}(x_i|u_i,s_i)\\
&= \prod_{i=1}^n p_{\rvy_{e}|\rvu,\rvs}(y_{e,i}|u_i,s_i)\\
\end{align*}
The second step above follows from the fact that the channel is
memoryless and the symbol $\rvx_i$ at time $i$ is generated as a
function of $(\rvu_i,\rvs_i)$. Hence we have that
\begin{align}
&\frac{1}{n}H(\rvye^n|\rvs^n,\rvu^n) %= \sum_{i=1}^n
%H(\rvy_{e,i}|\rvs^n,\rvu^n,\rvy_{e,1}^{i-1})\\
= \sum_{i=1}^n H(\rvy_{e,i}|\rvs_i,\rvu_i).\label{eq:EquivTerm2}
\end{align}

Furthermore note that
\begin{equation}
\frac{1}{n}H(\rvye^n) \le \sum_{i=1}^n H(\rvy_{ei}).
\label{eq:YeBound}\end{equation}

Finally, in order to lower bound the term $I(\rvs^n;\rvye^n|\rvu^n)$
we let $J$ to be a random variable which equals 1 if
$(\rvs^n,\rvu^n)$ are jointly typical. Note that $\Pr(J=1) = 1 -
o_n(1)$.
\begin{align}
&\frac{1}{n}I(\rvs^n;\rvye^n|\rvu^n) = \frac{1}{n}H(\rvs^n|\rvu^n)
- \frac{1}{n}H(\rvs^n|\rvu^n,\rvye^n) \notag\\
&\ge \frac{1}{n}H(\rvs^n|\rvu^n,J=1)\Pr(J=1)
- \frac{1}{n}H(\rvs^n|\rvu^n,\rvye^n)\notag\\
&\ge \frac{1}{n}H(\rvs^n|\rvu^n,J=1) -
\frac{1}{n}H(\rvs^n|\rvu^n,\rvye^n) - o_n(1) \notag\\
&\ge H(\rvs|\rvu) - \frac{1}{n}H(\rvs^n|\rvu^n,\rvye^n) - o_n(1)
\label{eq:NoMem}\\
&\ge H(\rvs|\rvu) -
\frac{1}{n}\sum_{i=1}^nH(\rvs_i|\rvu_i,\rvy_{e,i}) -
o_n(1)\label{eq:mutInfBound}\end{align} where~\eqref{eq:NoMem}
follows from the fact that $\rvs^n$ is an i.i.d.\ sequence and hence
conditioned on the fact that $(\rvs^n,\rvu^n)$ is a pair of typical
sequence there are $2^{n H(\rvs|\rvu)- no_n(1)}$ possible sequences
$\rvs^n$.

Substituting~\eqref{eq:EquivTerm1},~\eqref{eq:EquivTerm2},~\eqref{eq:YeBound}
and~\eqref{eq:mutInfBound} in the lower bound~\eqref{eq:condEntLB}
and using the fact that as $n \rightarrow \infty$, the summation
converges to the mean values,
\begin{align*}
&\frac{1}{n}H(\rvk|\rvye^n) \\
 &= I(\rvu;\rvyr)+ H(\rvye|\rvu, \rvs)- H(\rvye)+ H(\rvs|\rvu)\!-
\! H(\rvs|\rvu,\rvye) \!- \!o_n(1)\!\\
 &= I(\rvu;\rvyr)-I(\rvye;\rvs|\rvu) - I(\rvye;\rvu)
 +I(\rvye;\rvs|\rvu)-o_n(1)\\
 &= I(\rvu;\rvyr) - I(\rvye;\rvu) - o_n(1)
\end{align*}as required.

\subsection{Proof of Theorem~\ref{thm:ubNoDisc}}
A sequence of length-$n$ code satisfies:
\begin{align}
&\frac{1}{n}H(\rvk|\rvyr^n) \le \eps_n~\label{eq:fano}\\
&\frac{1}{n}H(\rvk|\rvye^n) \ge \frac{1}{n}H(\rvk)-
\eps_n~\label{eq:secrecy}
\end{align}
where~\eqref{eq:fano} follows from the Fano's inequality since the
receiver is able to recover the secret-key $\rvk$ given $\rvyr^n$
and~\eqref{eq:secrecy} is a consequence of the secrecy constraint.
Furthermore, note that $\rvk \rightarrow (\rvx^n,\rvs^n) \rightarrow
(\rvyr^n,\rvye^n)$ holds as the encoder generates the secret key
$\rvk$. Thus we can bound the rate $R = \frac{1}{n}H(\rvk)$ as
below:
\begin{align}
nR &\le I(\rvk;\rvyr^n|\rvye^n) + 2n\eps_n \notag\\
&\le I(\rvk,\rvs^n,\rvx^n;\rvyr^n|\rvye^n) + 2n\eps_n \notag\\
&\le I(\rvs^n,\rvx^n;\rvyr^n|\rvye^n) + H(\rvk|\rvs^n,\rvx^n) + 2n\eps_n \notag\\
&= I(\rvs^n,\rvx^n;\rvyr^n|\rvye^n) + 3n\eps_n \label{eq:Fano_k}\\
&\le \sum_{i=1}^n I(\rvs_i,\rvx_i;\rvy_{r,i}|\rvy_{e,i}) + 3n\eps_n \label{eq:memChannel}\\
&\le nI(\rvx,\rvs;\rvyr|\rvye) + 3n\eps_n
\end{align}
where~\eqref{eq:Fano_k} follows from the Fano Inequality because $\rvk$ can be obtained from $(\rvx^n,\rvs^n)$,~\eqref{eq:memChannel} from from the fact that the channel is memoryless and the last step follows from the concavity of the conditional
entropy term $I(\rvx,\rvs;\rvyr|\rvye)$ in the input distribution
$p_{\rvx,\rvs}$ (see e.g.,~\cite{khistiTchamWornell:07}).

Finally since the secret-key capacity only depends on the marginal
distribution of the channel and not on the joint distribution we can
minimize over all joint distributions with fixed marginal
distributions.

\section{Gaussian Case}
\begin{figure}
\begin{center}
\includegraphics[scale=0.4]{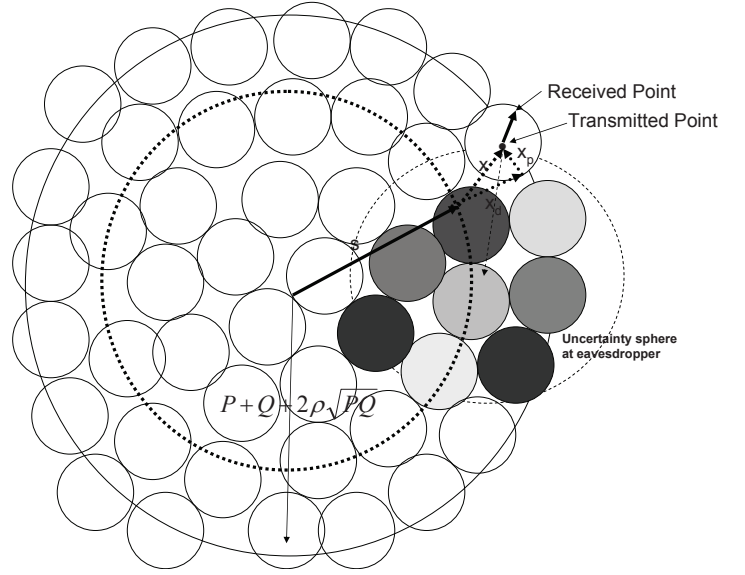}
\end{center}\caption{Secret-key agreement codebook for the dirty paper channel. The transmit sequence $\rvx^n$ is selected so that $\rvu^n = \rvx^n+\rvs^n$ is a sequence in the codebook $\cC$.   The smaller spheres above denote the noise uncertainty at the legitimate receiver. Their centres are the codewords in $\cC$. The larger sphere denotes the noise uncertainty at the eavesdropper. Our binning of smaller spheres  guarantees that the noise uncertainty sphere of the eavesdropper has all possible messages, resulting in (asymptotically) perfect equivocation.}
\label{fig:dpc}\end{figure}

We develop the lower and upper bounds on secret-key agreement capacity for the Gaussian channel model.
\subsection{Proof of Prop.~\ref{prop:lbGaussNoDisc}}
Recall that $\rvs \sim \cN(0,Q)$. Choose  $\rvx \sim \cN(0,P)$ to be a Gaussian random variable
independent of $\rvs$  and let $E[\rvx\rvs] = \rho \sqrt{PQ}$. Select $\rvu = \rvx + \alpha \rvs$ and the lower bound follows by evaluating
\begin{align*}
R &= I(\rvu;\rvyr) - I(\rvu;\rvye) \\
&= h(\rvu|\rvye) - h(\rvu|\rvyr)\end{align*}
Further evaluating each of the terms above with $\rvu = \rvx + \al \rvs$, note that
\begin{align*}&h(\rvu|\rvye)= h(\rvx  + \al \rvs |\rvx + \rvs + \rvze)= \\ & \frac{1}{2}\log2\pi e\left(P + \al^2 Q + 2\al\rho \sqrt{PQ} \right. - \\&\qquad\left.
\frac{(P+\al Q + (1+\al)\rho\sqrt{PQ})^2}{P+Q+1+\Delta + 2\rho\sqrt{PQ}}\right)\end{align*}and\begin{align*}
&h(\rvu|\rvyr)= h(\rvx + \al \rvs | \rvx + \rvs + \rvzr)=\\ &\quad \frac{1}{2}\log2\pi e\left(P + \al^2 Q + 2\al\rho\sqrt{PQ} \right.- \\&\qquad \left.
\frac{(P+\al Q + \rho(1+\al)\sqrt{PQ})^2}{P+Q+1+2\sqrt{PQ}}\right).\end{align*}
\iffalse \begin{align*}&= \frac{1}{2}\log\left(\frac{PQ(\al-1)^2(1-\rho^2) +
(1+\Delta)( P+ \al^2 Q)+ 2\rho\al\sqrt{PQ}(1+\Delta)}{P+Q+1+\Delta + 2\rho\sqrt{PQ}}\right) \\
&\quad - \frac{1}{2}\log\left(\frac{PQ(\al-1)^2(1-\rho^2) +
( P+ \al^2Q)+ 2\rho\al\sqrt{PQ}}{P+Q+1 + 2\rho\sqrt{PQ}}\right) \\
&= \frac{1}{2}\log\left(1 + \frac{\Delta(P+ \al^2Q +2\rho\al\sqrt{PQ})}{P+ \al^2Q + 2\rho\al\sqrt{PQ}
+PQ(\al-1)^2(1-\rho^2)}\right)\\&\quad +
\frac{1}{2}\log\left(\frac{P+Q+1 + 2\rho\sqrt{PQ}}{P+Q+1+\Delta+2\rho\sqrt{PQ}}\right)\\
&= \frac{1}{2}\log\left(1 + \frac{\Delta}{1
+\frac{PQ(\al-1)^2(1-\rho^2)}{P+ \al^2Q +2\rho\al\sqrt{PQ}}}\right)\\&\quad +
\frac{1}{2}\log\left(\frac{P+Q+1 + 2\rho\sqrt{PQ}}{P+Q+1+\Delta+2\rho\sqrt{PQ}}\right)
\end{align*}\fi

This yields that
\begin{multline}
R= \frac{1}{2}\log\left(1 + \frac{\Delta}{1
+\frac{PQ(\al-1)^2(1-\rho^2)}{P+ \al^2Q +2\rho\al\sqrt{PQ}}}\right)\\
\quad +\frac{1}{2}\log\left(\frac{P+Q+1 + 2\rho\sqrt{PQ}}{P+Q+1+\Delta+2\rho\sqrt{PQ}}\right).
\end{multline}
Note that the first term in the expression above is maximized when $\al =1$. In this case we have that
\begin{align}
R &= \frac{1}{2}\log\left(\frac{(1 + \Delta)(P+Q+1 + 2\rho\sqrt{PQ})}{P+Q+1+\Delta+2\rho\sqrt{PQ}}\right) \\
&= \frac{1}{2}\log\left( 1 + \frac{\Delta (P+Q  + 2\rho\sqrt{PQ})}{P+Q+1+\Delta+2\rho\sqrt{PQ}}\right)
\end{align}
as required.

To complete the proof we show that the choice $\al = 1$ is indeed feasible when $P\ge 1$ and 
$(P,\rho)$ satisfy \eqref{eq:rhoCons}.

In particular the constraint~\eqref{eq:lowerBoundNoDiscCons} requires
that
\begin{align*}
&h(\rvu|\rvs) \ge h(\rvu|\rvyr)\\
&\Rightarrow h(\rvx |\rvs) \ge h(\rvx + \rvs|\rvx + \rvs + \rvzr)\\
& \Rightarrow \frac{1}{2}\log P(1-\rho^2) \ge 
\frac{1}{2}\log\left(\frac{P+ Q + 2\rho\sqrt{PQ})}{P+Q+1+2\rho\sqrt{PQ}}\right).
\end{align*}Rearranging,
\begin{align}
P(1-\rho^2) &\ge 1-\frac{1}{P+Q+1 + 2\rho\sqrt{PQ}} \ge 1-\frac{1}{P+Q+1}
\end{align}
as required.

It is worth comparing the choice of the auxiliary variable $\rvu = \rvx + \rvs$ in the present problem with the choice of optimal $\rvu$
in the dirty paper coding problem~\cite{costa:83}. While the input $\rvx$ is independent of $\rvs$ in~\cite{costa:83}, as illustrated in Fig.~\ref{fig:dpc} the 
optimal $\rvx$ in the secret-key problem has a component along $\rvs$. This is because scaling the interference sequence increases the secret-key rate.
Secondly recall that in~\cite{costa:83} we find the auxiliary codeword $\rvu^n$ that is closest to $\al \rvs^n$ where $\al = \frac{P}{P+N}$. In contrast this MMSE scaling is not performed in the secret-key problem.

\subsection{Proof of Prop.~\ref{prop:ubGaussNoDisc}}

We evaluate the upper bound in Theorem~\ref{thm:ubNoDisc} for the
choice $\rvze = \rvzr + \rvz_\delta$, where $\rvz_\delta \sim
\cN(0,\Delta)$ is independent of $\rvzr$.

\begin{align*}
&I(\rvs,\rvx;\rvyr|\rvye) = h(\rvyr|\rvye)-
h(\rvyr|\rvye,\rvx,\rvs)\\
&= h(\rvyr|\rvye)-
h(\rvzr|\rvze)\\
&\le \frac{1}{2}\log\left(P+Q+1 + 2\sqrt{PQ}-
\frac{(P+Q+1 + 2\sqrt{PQ})^2}{P+Q+1+\Delta + 2\sqrt{PQ}}\right)-
\\&\quad-\frac{1}{2}\log\left(1- \frac{1}{1+\Delta}\right)
\end{align*}
where we have used the fact that the conditional entropy
$h(\rvyr|\rvye)$ is maximized by a Gaussian distribution~\cite{khistiWornell:MISOME}. The above
expression gives~\eqref{eq:ubGaussNoDisc}.

\section{Symmetric CSI}

We establish the secret-key capacity for the case of symmetric channel state information  i.e., when $\rvs^n$ is revealed to both the transmitter and the legitimate receiver.

\subsection{Achievability for Theorem~\ref{thm:symCSI}}
As explained in section~\ref{subsec:CSI} the achievability result follows directly from Theorem~\ref{thm:lbNoDisc} by replacing $\rvyr$ with $\orvyr = (\rvyr,\rvs)$ in the lower bound expression. We nevertheless provide an alternate scheme that only requires the knowledge of causal CSI at the transmitter. The idea is to use a different wiretap codebook for each realization of the state variable. In particular suppose that $\cS = \{s_1,\ldots, s_M\}$ denote the set of available states. Since the encoder and the decoder are both aware of the state realization $\rvs_i$ and  can use this common knowledge to select the appropriate codebook for transmission.  These codebooks are  constructed assuming that the eavesdropper is also revealed the state. Suppose that we fix the distribution $p_{\rvu,\rvx|\rvs=\rvs_i}(\cdot)$ in~\eqref{eq:symCap}. Let \begin{equation}\label{eq:Ri}R_i = I(\rvu;\rvyr|\rvs = s_i)-I(\rvu;\rvye|\rvs=s_i)\end{equation} and $p_i = \Pr(\rvs = s_i)$.  For each $i=1,2\ldots, M$, a wiretap codebook of length $np_i$ and rate $R_i$ is constructed and used to transmit a message $\rvk_i$. Another independent key $\rvk_s$ of rate $R_s = H(\rvs|\rvye)$ is then generated by exploiting the fact that $\rvs^ n$ is not known to the eavesdropper.

\subsubsection{Codebook Construction}

\begin{itemize}
\item For each $i=1,\ldots, M$ generate a codebook $\cC_i$ of rate $R_i-2\eps_n$ and length $n_i = n(p_i-\eps_n)$ by sampling the codewords i.i.d.\ from the distribution $p_{\rvu|\rvs}(\cdot|s_i)$. 
\item Construct a codebook $\cC_s$ where the set of all typical sequences $\rvs^n$ of size $2^{n(H(s)-2\eps_n)}$ is partitioned into $2^{n(R_s -\eps_n)}$ bins each containing $2^{n(I(\rvs;\rvye)-\eps_n)}$ sequences. 
\end{itemize}

\subsubsection{Encoding}
\begin{itemize}
\item For each $i=1,\ldots, M$ the transmitter selects a random message $\rvk_i$ and a random codeword sequence $\rvt_i^{n_i}$ in the corresponding in  the corresponding bin of $\cC_i$. 
\item Upon observing $\rvs(j) = s_i$ at time $t=j$, it selects the next available symbol of $\rvt_i^{n_i}$ and samples the channel input symbol from the distribution $p_{\rvx|\rvs,\rvu}$.
\item At the end of the transmission it looks for the bin index of $\rvs^n$ in $\cC_s$ and declares this to be $\rvk_s$.
\item The overall secret-key is $(\rvk_1,\ldots, \rvk_M,\rvk_s)$.
\end{itemize}

\subsubsection{Decoding}
\begin{itemize}
\item The decoder divides $\rvyr^n$ into subsequences $(\rvy_{1}^{n_1},\ldots, \rvy_{M}^{n_M})$, where the subsequences $\rvy_i^{n_i}$ is obtained by collecting the symbols of $\rvyr^n$ when $\rvs = s_i$.
\item For $i=1,\ldots, M$ it searches for a codeword $\rvt_i^{n_i}$ in $\cC_i$ that is jointly typical with $\rvy_i^{n_i}$. If no such codeword or multiple codewords is found an error is declared. Otherwise the bin index of $\rvt_i^{n_i}$ is taken as declared as the message $\hat{\rvk}_i$.
\end{itemize}

Through standard arguments it can be shown that the error probability in decoding at the legitimate receiver vanishes as $n\rightarrow\infty$ provided we select the rates according to~\eqref{eq:Ri}. We omit the details due to space constraints.

\subsubsection{Secrecy Analysis}
First, consider splitting $\rvye^n = (\rvy_{e1}^{n_1},\ldots, \rvy_{eM}^{n_M})$ where the subsequence $\rvy_{ej}^{n_j}$ is obtained by grouping the symbols of $\rvye^n$ when $\rvs = s_j$. From the construction of the wiretap codebook $\cC_j$ it follows that\begin{equation}\label{eq:WTequiv1}\frac{1}{n}H(\rvk_j|\rvy_{ej}^{n_j}) \ge \frac{1}{n}H(\rvk_j)-\eps_n,\qquad j=1,\ldots, M\end{equation}
Next since the messages are selected independently and the encoding functions are also independent it follows that
\begin{align}
&~\frac{1}{n}H(\rvk_j|\rvk_1,\ldots, \rvk_{j-1},\rvk_{j+1},\ldots, \rvk_M, \rvye^n,\rvs^n) \notag\\
&=\frac{1}{n}H(\rvk_j|\rvy_{ej}^n) \ge \frac{1}{n}H(\rvk_j) - \eps_n \label{eq:WTequiv2}
\end{align}
Thus by the chain rule we have that
\begin{equation}
\frac{1}{n}H(\rvk_1,\ldots, \rvk_M|\rvye^n,\rvs^n) \ge R_0 - \eps_n \label{eq:ChContr}
\end{equation}
where $R_0 = H(\rvk_1,\ldots, \rvk_M) = I(\rvu;\rvyr|\rvs)-I(\rvu;\rvye|\rvs)$. To complete the secrecy analysis we require the following additional result
\begin{lemma}
For any input distribution $p_{\rvu,\rvx|\rvs}$ such that $I(\rvu;\rvyr|\rvs)> I(\rvu;\rvye|\rvs)$ we have that
\begin{equation}
\frac{1}{n}H(\rvs^n|\rvye^n)\ge \frac{1}{n}H(\rvs|\rvye)- o_n(1).\label{eq:SrcContr}
\end{equation}
\end{lemma}
\begin{proof}
First observe that we can write:
\begin{align}
&\frac{1}{n}H(\rvs^n|\rvye^n)= \frac{1}{n}H(\rvye^n|\rvs^n) + \frac{1}{n}H(\rvs^n)- \frac{1}{n}H(\rvye^n)\\
&=\frac{1}{n}H(\rvye^n|\rvs^n,\rvu^n) + \frac{1}{n}I(\rvu^n;\rvye^n|\rvs^n) + \frac{1}{n}H(\rvs^n)-\frac{1}{n}H(\rvye^n).\label{eq:expanded}
\end{align}
We now observe the following. Since the channel from $(\rvu^n,\rvs^n)\rightarrow \rvye^n$ is memoryless, 
\begin{equation}
\frac{1}{n}H(\rvye^n|\rvs^n,\rvu^n) = \frac{1}{n}\sum_{i=1}^n H(\rvy_{ei}|\rvs_i,\rvu_i) \rightarrow H(\rvye|\rvs,\rvu)\label{eq:t1}
\end{equation}as $n\rightarrow \infty$. Next note that by construction
\begin{align}
\frac{1}{n}H(\rvu^n|\rvs^n) = I(\rvu;\rvyr|\rvs)-2\eps_n,
\end{align}
and since  $I(\rvu;\rvyr|\rvs)> I(\rvu;\rvye|\rvs)$ it follows  (c.f.~\cite[Lemma~1]{khiangElGamal:09}) that\footnote{Intuitively for any typical $\rvs^n$, the total number of sequences $\rvu^n$  is $2^{n I(\rvu;\rvyr|\rvs)}$. The probability that a sequence $\rvu^n$ is jointly typical with $\rvye^n$ is $2^{-nI(\rvu;\rvye|\rvs)}$. A precise argument involves bounding the expected size of the list and invoking a concentration result. }
\begin{align}
\frac{1}{n}H(\rvu^n|\rvs^n,\rvye^n) \le I(\rvu;\rvyr|\rvs) - I(\rvu;\rvye|\rvs) - o_n(1)
\end{align}
Combining the above two inequalities,
\begin{align}
\frac{1}{n}I(\rvu^n;\rvye^n|\rvs^n) \ge I(\rvu;\rvye|\rvs) - o_n(1)\label{eq:t2}
\end{align}
Since the sequence $\rvs^n$ is sample i.i.d.\ we have
\begin{align}
\frac{1}{n}H(\rvs^n) = H(\rvs) \label{eq:t3}
\end{align} and finally from the chain rule
\begin{align}
\frac{1}{n}H(\rvye^n) &\le \frac{1}{n}H(\rvy_{ei}) \rightarrow H(\rvye)\label{eq:t4}
\end{align} as $n\rightarrow \infty$. Substituting~\eqref{eq:t1},~\eqref{eq:t2},~\eqref{eq:t3} and~\eqref{eq:t4} into~\eqref{eq:expanded} completes the claim.
\end{proof}
The secrecy analysis can be completed by combining~\eqref{eq:ChContr} and~\eqref{eq:SrcContr} as shown below.
\begin{align}
&~\frac{1}{n}H(\rvk_1^M, \rvk_s|\rvye^n) = \frac{1}{n}H(\rvk_1^M |\rvk_s, \rvye^n) + \frac{1}{n}H(\rvk_s|\rvye^n)\\
&\ge \frac{1}{n}H(\rvk_1^M |\rvs^n, \rvye^n) + \frac{1}{n}H(\rvk_s|\rvye^n) \label{eq:condEnt}\\
&\ge I(\rvu;\rvyr|\rvs)-I(\rvu;\rvye|\rvs) + \frac{1}{n}H(\rvk_s|\rvye^n) - o_n(1) \label{eq:chContr2}\\
&\!\ge\! I(\rvu;\rvyr|\rvs)\!-\!I(\rvu;\rvye|\rvs) \!+ \!\frac{1}{n}H(\rvs^n|\rvye^n)-\frac{1}{n}H(\rvs^n|\rvye^n,\rvk_s)\!-\!o_n(1) \label{eq:condEnt2}\\
&\ge I(\rvu;\rvyr|\rvs)-I(\rvu;\rvye|\rvs) + H(\rvs|\rvye)-\frac{1}{n}H(\rvs^n|\rvye^n,\rvk_s)-o_n(1) \label{eq:SrcContr2}\\
&= I(\rvu;\rvyr|\rvs)-I(\rvu;\rvye|\rvs) + H(\rvs|\rvye)-o_n(1) \label{eq:FanoDecoding}
\end{align}where~\eqref{eq:condEnt} and~\eqref{eq:condEnt2} follow from the fact that $\rvk_s$ is a deterministic function of $\rvs^n$ while~\eqref{eq:chContr2} follows by substituting~\eqref{eq:ChContr} and~\eqref{eq:SrcContr2} follows by substituting~\eqref{eq:SrcContr} while~\eqref{eq:FanoDecoding} follows from the fact that $\frac{1}{n}H(\rvs^n|\rvye^n,\rvk_s) \rightarrow 0$ as $n\rightarrow\infty$, since from the construction of $\cC_s$ there are at-most $2^{n (I(\rvs;\rvye)-\eps_n)}$ sequences associated with any given bin. Hence  the decoder can decode $\rvs^n$ with high probability and hence Fano's inequality applies.

\subsection{Converse}
\label{subsec:converse}
For any sequence of codes indexed by the codeword length $n$, we show that the secret key rate is upper bounded by the
capacity expression~\eqref{eq:symCap} plus a term that vanishes to zero as
the block length goes to zero. By applying the Fano inequality on the secret-key rate, we have that for some sequence $\eps_n$ that
approaches zero as $n$ goes to infinity that
\begin{align}
nR \le I(\rvk;\rvl) + n\eps_n \le I(\rvk;\rvs^n,\rvyr^n) + n\eps_n
\end{align}
where the last step follows from the data processing inequality
since $\rvl = h_n(\rvs^n, \rvyr^n)$. Furthermore from the secrecy condition $I(\rvk;\rvye^n) \le n\eps_n$ and hence,
\begin{align}
&nR \le I(\rvk;\rvs^n,\rvyr^n) - I(\rvk;\rvye^n) + 2n\eps_n\\
&\le \sum_{i=1}^n I(\rvk;\rvy_{ri},\rvs_i|\rvye^{i-1}\rvy_{r,i+1}^n,\rvs_{i+1}^n) - I(\rvk;\rvy_{e,i}|\rvye^{i-1}\rvy_{r,i+1}^n,\rvs_{i+1}^n),
\end{align}where the second step follows from the Csiszar sum-identity~\cite[Chapter 2]{ElGamalKim}
applied to difference of mutual informations. The derivation is analogous to~\cite{csiszarKorner:78} and is omitted. If we let $\rvv_i = (\rvye^{i-1}\rvy_{r,i+1}^n,\rvs_{i+1}^n)$ and $\rvu_i = (\rvk,\rvv_i)$ note that $\rvv_i \rightarrow \rvu_i \rightarrow (\rvx_i,\rvs_i) \rightarrow (\rvy_{r,i},\rvy_{e,i})$ holds.  Maximizing over each term in the summation we obtain that
\begin{align}
R &\le \max_{p_{\rvu,\rvv,\rvx}} I(\rvu;\rvyr,\rvs|\rvv)- I(\rvu;\rvye|\rvv) + 2\eps_n \label{eq:st1}\\
&= \max_{p_{\rvu,\rvx}} I(\rvu;\rvyr,\rvs)- I(\rvu;\rvye) + 2\eps_n \label{eq:st2}
\end{align}
where the second step follows from the fact that the maximizing over $\rvv$ is redundant since~\eqref{eq:st1} involves a convex combination  of $I(\rvu;\rvyr,\rvs|\rvv=v_i) - I(\rvu;\rvye|\rvv=v_i)$ and hence we can replace with the term that results in the largest value.  We recover~\eqref{eq:symCap} from~\eqref{eq:st2} by using an approach similar to~\eqref{eq:ft2}.

\section{Conclusions}

We study the secret key agreement capacity over a wiretap channel controlled by a state parameter. Lower and upper bounds on the capacity are established when the state sequence is known noncausally to the encoder. The lower bound is obtained by creating a common reconstruction sequence at the legitimate terminals and binning the set of reconstruction sequences to generate a secret key. When evaluated for the Gaussian case (secret-key from dirty paper) our bounds coincide in the high SNR and high INR regimes and the gap between the two bounds is always less than 0.5 bits.  We also observe that the rates for secret-key agreement can be significantly higher than that proposed for the secret message transmission problem. We also extend our earlier~\cite{khisti:09} results on symmetric CSI to the general case of asymmetric CSI. 

A complete characterization of the secret-key capacity is obtained for the case of symmetric channel state information i.e., when the state sequence is known to both the encoder and the decoder. In this case we also present another coding scheme that involves multiplexed wiretap codebooks and only requires causal knowledge of the state sequence at the encoder. The capacity expression also captures an interesting tradeoff between correlating the input with the state sequence to maximize the contribution of the wiretap codebook and masking the state sequence from the eavesdropper. We illustrate this with a numerical example.

In terms of future work it will be interesting to study the secret key agreement capacity when there is only causal state information available to the transmitter. While this paper establishes the capacity when there is symmetric CSI at both the legitimate terminals, the more general problem of two-sided CSI remains to be explored. In another direction, secret key agreement protocols also appear to be an important component in more general problems. For example in~\cite{prabhakaran:09} the authors independently developed a secret-key agreement scheme as a building block in characterizing a secret message and secret key tradeoff for wiretap channels with correlated sources.   Another recent work~\cite{khiangElGamal:10} studies the problem of secret message transmission on wiretap channel with symmetric CSI and uses a block Markov encoding scheme that generates a secret key in each block~\cite{khisti:09}. Exploring such connections is an interesting direction  for future research. 

\bibliographystyle{IEEEtran}
\bibliography{sm}
\end{document}